\documentclass[uslettersize, 10pt, conference]{ieeeconf}      % Use this line for a4 paper

\IEEEoverridecommandlockouts                              % This command is only needed if 
\usepackage[OT1]{fontenc} 
\usepackage[english]{babel}
\usepackage{cite}
\usepackage[left=19.1mm,right=19.1mm,top=19.1mm,bottom=19.1mm]{geometry}
\usepackage{amsmath,amssymb,amsfonts}
\usepackage{textcomp}
\usepackage{xcolor}
\usepackage{blindtext}
\usepackage{subfiles} 
\usepackage{mathtools}
\usepackage{adjustbox}
\usepackage{tikz}
\usepackage{svg}
\usepackage{booktabs} % for rulers in tables
\usepackage[final]{microtype} % for \microtypesetup{protrusion=true}
\usepackage{algorithm} 
\usepackage[algo2e,ruled,vlined]{algorithm2e} 
\usepackage{hyperref}
\hypersetup{hidelinks, breaklinks = true}
\usepackage[capitalise]{cleveref}
\usepackage{mathdots}
\usepackage{yhmath}
\usepackage{cancel}
\usepackage{color}
\usepackage{array}
\usepackage{multirow}
\usepackage{gensymb}
\usepackage{adjustbox}
\usepackage{tabularx}
\usepackage{comment}
\usepackage{extarrows}
\usepackage{booktabs}
\usepackage{diagbox}
\usepackage{graphicx}
\usepackage{subcaption}
\usepackage[normalem]{ulem}
\useunder{\uline}{\ul}{}
\usetikzlibrary{fadings}
\usetikzlibrary{patterns}
\usetikzlibrary{shadows.blur}
\usetikzlibrary{shapes}
\graphicspath{ {./Figures/} }
\usepackage{amsmath}
\usepackage{tikz}
\usepackage{mathdots}
\usepackage{yhmath}
\usepackage{cancel}
\usepackage{color}
\usepackage{array}
\usepackage{multirow}
\usepackage{amssymb}
\usepackage{bm}
\usepackage{gensymb}
\usepackage{tabularx}
\usepackage{extarrows}
\usepackage{booktabs}
\usetikzlibrary{fadings}
\usetikzlibrary{patterns}
\usetikzlibrary{shadows.blur}
\usetikzlibrary{shapes}
\usepackage{algcompatible}
\SetKwComment{Comment}{/* }{ */}
\newtheorem{thm}{Theorem}[section]

\newtheorem{defn}{Definition}[section]

\newtheorem{rem}{Remark}

\title{\LARGE \bf
A Cooperative Compliance Control Framework for Socially Optimal Mixed Traffic Routing

}

\author{Anni Li, Ting Bai, Yingqing Chen,  Christos G. Cassandras and Andreas A. Malikopoulos % <-this % stops a space
\thanks{This work was supported in part by NSF under grant CNS-2149511, CNS-2401007, CMMI-2219761, IIS-2415478, Mathworks, and by ARPAE under grant DE-AR0001282.}% <-this % stops a space
\thanks{ A. Li, Y. Chen, and C. G. Cassandras are
 with the Division of Systems Engineering and Center for Information and
 Systems Engineering, Boston University, Brookline, MA 02446
 (email:\{anlianni; yqchenn; cgc\}@bu.edu).}
 \thanks{T. Bai and A. A. Malikopoulos are with the Information and Decision Science Lab, School of Civil $\&$ Environmental Engineering, Cornell University, Ithaca, New York, U.S.A. E-mails: \{{tingbai, amaliko\}@cornell.edu}}
}

\begin{document}

\maketitle
\thispagestyle{empty}
\pagestyle{empty}

%%%%%%%%%%%%%%%%%%%%%%%%%%%%%%%%%%%%%%%%%%%%%%%%%%%%%%%%%%%%%%%%%%%%%%%%%%%%%%%%
\begin{abstract}
In mixed traffic environments, where Connected and Autonomed Vehicles (CAVs) coexist with potentially non-cooperative Human-Driven Vehicles (HDVs), the self-centered behavior of human drivers may compromise the efficiency, optimality, and safety of the overall traffic network. In this paper, we propose a Cooperative Compliance Control (CCC) framework for mixed traffic routing, where a Social Planner (SP) optimizes vehicle routes for system-wide optimality while a compliance controller incentivizes human drivers to align their behavior with route guidance from the SP through a “refundable toll” scheme. A key challenge arises from the heterogeneous and unknown response models of different human driver types to these tolls, making it difficult to design a proper controller and achieve desired compliance probabilities over the traffic network. To address this challenge, we employ Control Lyapunov Functions (CLFs) to adaptively correct (learn) crucial components of
our compliance probability model online,
%{\color{blue} Can we be more specific about what exactly we are learning in the model, e.g., specific parameters?}
%{\color{red} The term that is learned is $f_k$ in (20), and is adapted by using (22). Explaining this point here would require too many details. }
construct data-driven feedback controllers, and demonstrate that we can achieve the desired compliance probability for HDVs,
% through rigorous proofs
% {\color{blue} But we don't really have any such poofs here. We can either omit this or replace by "through a rigorous systematic control framework"}, 
thereby contributing to the social optimality of the traffic network.

\end{abstract}

%\begin{keywords}
 %Connected Autonomous Vehicles, Optimal Control, Mixed traffic, Iterated Best Response
%\end{keywords}

%%%%%%%%%%%%%%%%%%%%%%%%%%%%%%%%%%%%%%%%%%%%%%%%%%%%%%%%%%%%%%%%%%%%%%%%%%%%%%%%
\section{INTRODUCTION}
With rapid urbanization and population growth, traffic congestion and safety have long been major concerns for government organizations and society as a whole, making improving road network efficiency a key objective in intelligent transportation systems \cite{jeihani2022investigating,van2019travelers}. The Vehicle Routing Problem and its variants have gained significant popularity in research over the past decades \cite{toth2002vehicle}. Early approaches relied on static shortest path algorithms, such as Dijkstra's and Floyd-Warshall's algorithms \cite{hofner2012dijkstra}, which provide efficient routes based on predefined and deterministic road network conditions without accounting for dynamic traffic fluctuations. In a real-world traffic environment, information is always subject to uncertainties, such as actual demand, ongoing user requests, vehicle breakdowns, and traffic congestion. These fluctuations have motivated the development of the \emph{dynamic} vehicle routing problem \cite{abbatecola2016review,ritzinger2016survey}, which concurrently integrates real-time traffic data and re-optimizes the routes to adapt to varying conditions in route guidance by utilizing advanced sensing technologies, such as Roadside Units (RSUs), vehicle-to-vehicle (V2V), vehicle-to-infrastructure (V2I), and vehicle-to-smart terminal (V2T) communication \cite{jeihani2022investigating}. 

The emergence of Connected and Automated Vehicles (CAVs), also known as autonomous vehicles or self-driving cars, has the potential to significantly improve traffic performance and efficiency by better assisting drivers in making decisions to reduce accidents and traffic congestion\cite{bang2021AEMoD,Malikopoulos2020,xiao2021decentralized,li2025robust}. The shared information and cooperation among CAVs facilitate intelligent transportation management through precise trajectory planning and real-time traffic information. However, 100\% CAV penetration is unlikely to be achieved in the near future. This raises the question of how to coordinate CAVs with Human-Driven Vehicles (HDVs) to enhance traffic efficiency while also ensuring safety since human driving behavior is stochastic and hard to predict \cite{ghiasi2019mixed,wang2020controllability,li2024safe,xiao2024toward}. In such mixed traffic vehicle routing problems, HDVs generally operate based on individual driver preferences and limited traffic conditions, leading to \textit{selfish} optimal strategies where each vehicle selects the route that minimizes its own travel cost without considering the broader impact on the network. As a result, this decentralized decision-making will generally compromise system-wide efficiency and exacerbate traffic congestion \cite{roughgarden2005selfish,feldmann2003selfish}. In contrast, \textit{socially} optimal routing aims to minimize the system-wide travel cost by optimizing vehicle routes to balance network utilization and reduce congestion. This approach relies on centralized vehicle coordination, where a Social Planner (SP) collects complete information on the road network and assigns routes that improve the overall efficiency, even if they may not be individually optimal for every vehicle. The implementation of socially optimal routing in mixed traffic environments faces a significant challenge, as HDVs may deviate from assigned routes in pursuit of individual gains \cite{li2024towards,chen2008altruism,sadeghi2023social}. 

To overcome this challenge, there has been increasing research focused on integrating social psychology tools and implementing non-monetary management schemes into intelligent transportation systems to improve traffic efficiency and incentivize HDVs to behave in compliant ways. When solving traffic congestion problems, congestion pricing schemes are proposed in \cite{pigou2017economics,de2011traffic}, where individual drivers have to pay for negative externalities to ensure traffic efficiency. In \cite{kockelman2005credit}, credit-based congestion pricing (CBCP) is proposed, where drivers receive monetary travel credits to use on roads, while \cite{jalota2023credit} further adopts a mixed economy viewpoint to let eligible users receive travel credits while ineligible users pay out-of-pocket to use the express lanes. Considering the heterogeneity of human drivers, such congestion pricing schemes can be unfair \cite{brands2020tradable} as they tend to favor wealthier drivers. To overcome these limitations, non-tradable credit-based congestion pricing schemes have been applied as fair and equitable mechanisms in \cite{elokda2022carma}, which proposes a congestion management scheme called ``CARMA" utilizing non-tradable ``karma credits'' for all drivers to bid for the rights to use the express lanes; such credits may be redistributed so that the drivers can choose to benefit themselves now or in the future based on their individual value of time. Moreover, to deploy social contracts in practice, the use of Distributed Ledger Technologies (DLT) is described in \cite{ferraro2023personalised} to create personalized social nudges and to influence the behavior of users.

In this work, assuming HDVs can communicate with the SP, we propose a \emph{Cooperation Compliance Control} (CCC) framework to incentivize HDVs to comply with the route guidance provided by the SP. The control mechanism is implemented using a \emph{refundable toll} principle, where each vehicle commits a certain number of tokens to a digital wallet at the beginning of its driving period (which can be efficiently implemented through DLT). If the vehicle complies with the reference routes provided by the SP, the committed tokens will be fully refunded. Otherwise, a certain amount of tokens will be removed from the driver's digital wallet. Meanwhile, the SP continuously monitors the real-time system states, updating traffic flow information after a predetermined time window. Consequently, new reference routes will be re-assigned to non-compliant vehicles to mitigate traffic congestion and enhance the overall system-wide traffic efficiency.

\section{PROBLEM FORMULATION}
\label{sec:problem_formulation}
\subsection{Road Network}
Let $\mathcal{G}(\mathcal{V},\mathcal{E})$ be the directed graph induced by an urban road network, where $\mathcal{V}$ denotes a set of nodes and $\mathcal{E}\subset \mathcal{V}\times \mathcal{V}$ denotes a set of all edges in the network, and $|\cdot|$ is the cardinality of a set. For any node $i\in \mathcal{V}$ and the directed road edge $e^{ij}\in\mathcal{E}$ from node $i\in\mathcal{V}$ to node $j\in\mathcal{V}$, the social planner uses a function $s:\mathcal{E}\rightarrow \mathbb{R}^+$ as the travel cost for each edge depending on many factors, such as real-time states for other vehicles, road capacity, real traffic flow, travel time, speed limits, tolls, etc. In particular, $s^{ij}$ represents the travel cost for edge $e^{ij}\in\mathcal{E}$.

\subsection{Vehicle Routing Constraints}
The vehicle routing problem is essentially identifying the optimal set of routes for a set of vehicles to travel from their origin to their destination.
At any planning time $t$, let $\mathcal{K}(t)=\{1,2,..., N(t)\}$ denote a set that includes all $N(t)$ vehicles in the traffic network. We drop the time $t$ in this section to keep the notation simple.
%{\color{blue} Why only the start time? Shouldn't we update these every time a vehicle joins the system or leaves it?}
% {\color{blue}
 % We agreed this is time-varying. Should we write $K(t),N(t)$? Should we say that it is $K(t),N(t)$ but we'll drop the $t$ to keep notation simple?}
For each vehicle $k\in \mathcal{K}$, its origin and destination nodes are denoted as 
$O_k,D_k\in \mathcal{V}$, 
% {\color{blue} I propose we use subscripts when there is no need for both sub and super scripts: $O_k, D_k$}
respectively. Let the binary variable $x^{ij}_k\in\{0,1\}$  
% {\color{blue} Here we either use $x_k^{ij}$ or stick to $x_{ij}^k$. Either is OK}
denote the optimal route choice from node $i$ to $j$ traversed by vehicle $k$. Thus, $x^{ij}_k=1$ means that the edge $e^{ij}$ is included in the optimal route for $k$. 
To find the minimal travel cost for the vehicles in the road network, vehicle $k\in \mathcal{K}$ has to satisfy the following constraints.

\textbf{Flow Balance:} The number of times that vehicle $k\in \mathcal{K}$ enters the node $j\in\mathcal{V}/\{O_k,D_k\}$ is equal to the number of times it leaves, i.e.,
\begin{align}
\label{conseq:balanced_flow}
    \sum_{i=1}^{|\mathcal{V}|} x^{ij}_k=\sum_{i=1}^{|\mathcal{V}|} x^{ji}_k,\; \forall j\in\mathcal{V}/ \{O_k,D_k\}.
\end{align}

\textbf{Origin and Destination Determination:} The vehicle $k\in\mathcal{K}$ leaves its origin $O_k$ and enters the destination $D_k$ are ensured by the following two constraints, respectively:
\begin{align}
    \label{conseq:origin_flow}
    \sum_{j=1}^{|\mathcal{V}|} x_k^{O_k,j}=1,\;\;\sum_{i=1}^{|\mathcal{V}|} x_k^{i,D_k}=1.
\end{align}
The constraint \eqref{conseq:origin_flow} ensures that vehicle $k$ starts its trip from $O_k$ and ends at $D_k$.

\textbf{Redundant Visit Avoidance:} To avoid revisiting the same node in the road network $\mathcal{G}$, each node is limited to be visited at most once, i.e., 
\begin{align}
\label{conseq:node_visit_once}
    \sum_{j=1}^{|\mathcal{V}|} x^{ij}_k\leq 1, \forall i\in\{1,2,...,|\mathcal{V}|\}.
\end{align}

\textbf{Sub-tour Elimination:} To eliminate sub-tours in the solution, any edge ending at the origin $O_k$ will not be chosen in
vehicle $k$'s optimal route(s), and the flow is 0 if the two nodes of an edge are the same, i.e., 
\begin{align} \label{conseq:origin_flow_inverse}
    \sum_{i=1}^{|\mathcal{V}|} x_k^{i,O_k}=0,\;x_k^{ii}=0.
\end{align}

\subsection{Social Optimal Guidance}
Given a road network modeled as a directed graph $\mathcal{G}(\mathcal{V},\mathcal{E})$, a Social Planner (SP) seeks to enhance its performance by assigning each vehicle $k$ a socially optimal route, denoted as $R^{ref}_k$.
% {\color{blue} I definitely prefer $R_k^{ref} here!$}
The corresponding travel cost along the edge $e^{ij}\in\mathcal{E}$ is $s^{ij}_k$
% {\color{blue}We have the same issue as $x_{ij}^k$ here}
. In general, $R^{ref}_k$ can be optimized while improving mobility equity \cite{bang2024routing} to ensure equity and fairness over all vehicles in the traffic network, alleviating traffic congestion based on real-time traffic conditions, and minimizing travel time and energy consumption, etc, or trip equity using an appropriate equity metric as in \cite{bai2025routingguidanceemergingtransportation} for each vehicle when the travel information in the road environment is dynamically varying. Therefore, the travel cost $s^{ij}_k$ may differ across different problem settings. In this paper, we do not specify a particular form for it but instead, keep it in a general form at the discretion of the SP.
The SP aims to minimize the travel cost and determines the reference route $R^{ref}_k$ for each vehicle $k\in \mathcal{K}$ by solving the following optimization problem:
\begin{align}
\label{opt:social_guidance}
        J^{ref}_k:=&\min_{x^{ij}_k} \sum_{i=1}^{|\mathcal{V}|} \sum_{j=1}^{|\mathcal{V}|} s^{ij}_kx^{ij}_k,~~s.t. \;\;\eqref{conseq:balanced_flow}-\eqref{conseq:origin_flow_inverse}.
\end{align}
Once we obtain the optimal solution $x^{ij}_k$ from the optimization problem \eqref{opt:social_guidance}, the entire reference route $R^{ref}_k$ can be given by
\begin{align}
\label{refeq:origin}
    R^{ref}_k=\bigcup_{(i,j)\in\mathcal{E}} \{e^{ij}\ |\ x^{ij}_k=1\}.
\end{align}
%{\color{blue} You describe the route by its edges but later on in (12) we use nodes to define a route. Is there any reason why these are different? How about using nodes here as well?}

%{\color{red} Thanks, professor. The reason that we use edges to describe $R^{ref}_k$ is that we do not mention the order of the nodes for now, if we use nodes $(i,j)$ in (6), there may be a conflict between (6) and (8). Besides, the decision points (introduced in the following subsection) are defined based on the edges.}

\subsection{Human Preferences}
In a mixed-traffic environment, the travel cost $s^{ij}$ on edge $e^{ij}$ designed by the social planner may not be adopted by all Human-Driven Vehicles (HDVs). This divergence occurs because human drivers may either act selfishly by prioritizing only their personal travel time and distance or lack the necessary technology to communicate with other vehicles and access real-time traffic conditions. We define $V_{NC}(t)$, 
% {\color{blue} This is also time-varying and should follow what we do with $N(t), K(t)$}
as the set of non-compliant vehicles at time $t$ (for simplicity, we write $V_{NC}$ in the remainder of this subsection), consisting of all HDVs that do not follow the reference trajectory $R^{ref}_k$, $k\in \mathcal{K}$, their routes are given by the following optimization problem:
% {\color{blue} I would definitely use $J_k$ here}
\begin{align}
\label{opt:selfish_opt_init}
    J_k:=&\min_{x^{ij}_k} \sum_{i=1}^{|\mathcal{V}|}\sum_{j=1}^{|\mathcal{V}|}c^{ij}_k x^{ij}_k,~~ s.t. \;\;\eqref{conseq:balanced_flow}-\eqref{conseq:origin_flow_inverse}.
\end{align}
The only difference between the selfish optimal routing problem \eqref{opt:selfish_opt_init} and the socially optimal routing problem \eqref{opt:social_guidance} lies in their objectives, where $c^{ij}_k$ denotes the selfish travel cost issued by non-compliant vehicle $k\in V_{NC}$ along the edge $(i,j)\in\mathcal{E}$. Although we impose no constraints on the choice of $c^{ij}_k$, the default travel cost for HDVs is generally assumed to be the free flow time, which is commonly provided by widely used navigation tools like Google Maps. Since drivers can easily access this information, they often base their routing decisions on these estimated travel times, without considering potential congestion or network-wide optimization strategies. 
Due to the conflict between social and selfish goals, the optimal routes determined by problem \eqref{opt:selfish_opt_init} may differ from the reference routes given by \eqref{opt:social_guidance}
except for the special case where $c^{ij}_k=s^{ij}_k$.

%{\color{blue} You mentioned that in the simulations we use the free flow time as $c^{ij}_k$. To preempt a reviewer asking for an example of how this might be selected, I suggest we add here that a simple example or default travel cost is the free flow time as provided by commonly used tools like GoogleMaps that a user can easily look up.}

When non-compliant vehicles deviate from their socially optimal reference routes and choose to travel along the selfishly optimal routes, they may slow down the traffic flow on main roads and lead to congestion. To alleviate such traffic congestion and increase the probability of the non-compliant vehicles complying with the references provided by the SP, a dynamic routing problem is formulated and new references will be provided at their subsequent decision points following any deviation; at the same time, tolls will be deducted from the user's digital wallet to penalize non-compliant behaviors. 
Note that while the solution of (\ref{opt:selfish_opt_init}) is unknown to the SP, any deviation from $R_k^{ref}$ that vehicle $k$ makes is detectable at decision points where the edge choice made by $k$ is observable. This allows the SP to monitor non-compliance and adjust routing strategies accordingly.
More details about the determination of decision points and the process by which the SP reassigns new reference routes for non-compliant vehicles are discussed in the next section.

\section{SOCIAL-COMPLIANCE OPTIMAL PLANNER}
\label{sec:compliance_control}
In the dynamic routing problem, the travel cost $s^{ij}$ is updated in real time based on the actual traffic conditions monitored by the SP. This update process must account for the current number of vehicles in the road network while incorporating the effects of newly arriving and exiting vehicle flows.
% {\color{blue} This is why we need $K(t)$ to be time-varying. But an astute reader would have recognized this when $K$ is introduced, which is why I think we need to make this time dependence clear when $K$ is first defined!}
Both the set $\mathcal{K}(t)$ that includes all vehicles in the road network and the set $V_{NC}(t)$ that includes all non-compliant vehicles are updated each time a vehicle reaches a node and makes a routing decision.

\subsection{Decision Points}
In the routing problem, a decision point typically corresponds to an intersection or fork in the road network. 
Let $t_k^m$ 
% {\color{blue} I would prefer to use $t_k^m$, but I am open to hear reasons for why not}
be the time of vehicle $k$ arriving at the $m$-th decision point along the route from $O_k$ to $D_k$, and let $d_k^m(t_k^m) \in \mathcal{V}$ 
% {\color{blue} Maybe $d^m_k(t^m_k)$}
denote the actual $m$-th decision point detected by the SP at time $t_k^m$. Assuming the cardinality of the set $R^{ref}_k$ is $N^{ref}_k$,
% {\color{blue} Preferably $N_k^{ref}$}
the sequence of nodes along reference route $R^{ref}_k$ given by \eqref{refeq:origin} at the Origin $O_k$ with time $t^0_k$ can be written as
\begin{align}
    \Pi^{ref}_k=\{O_k,d^1_k(t^0_k),d^2_k(t^0_k),...,d^{N^{ref}_k(t^0_k)-1}_k(t^0_k),D_k\}.
\end{align}
% {\color{blue} Preferably use subscript $k$ throughout in the above}
\begin{figure}[hpbt]
    \centering   
    % \vspace*{-\baselineskip} 
    \includegraphics[ width=0.8\linewidth]{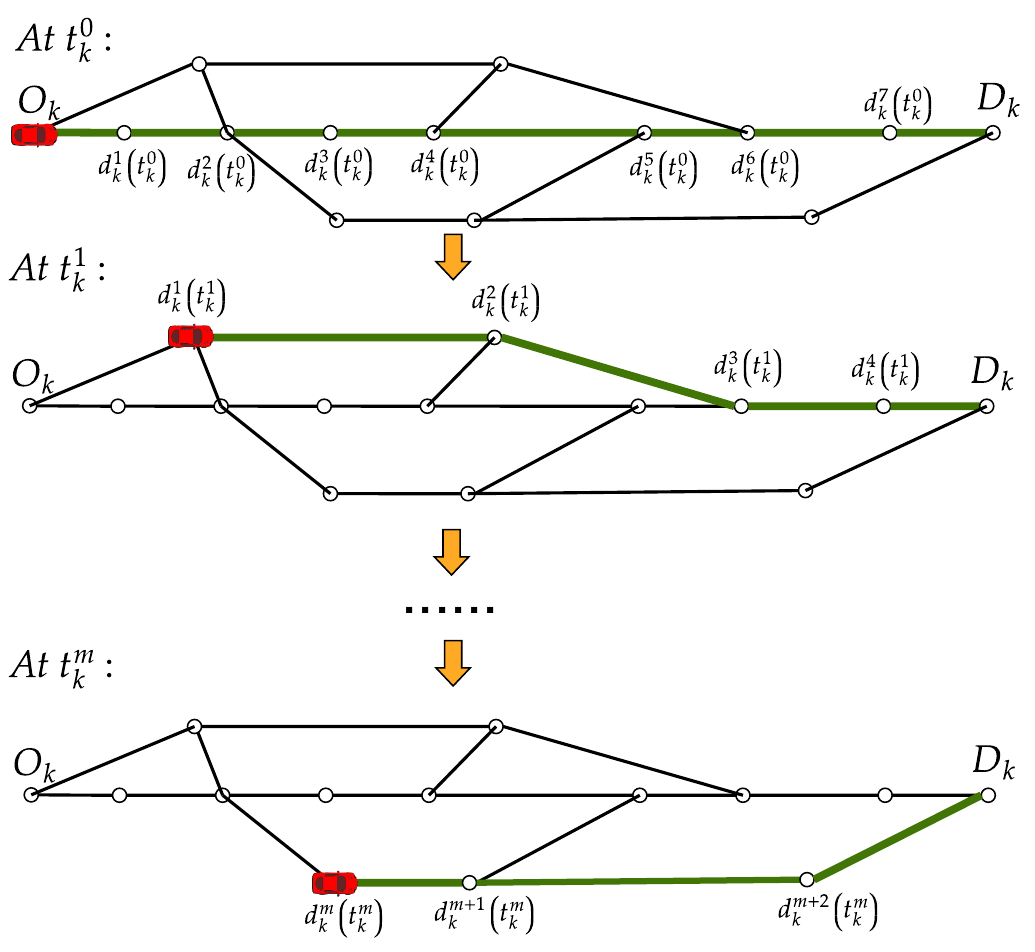} 
    % \vspace{-2mm}
    \caption{\small The order of decision points in a routing problem.}
    \label{fig:decision_pts}
   % \vspace*{-\baselineskip}
\end{figure}

As shown in Fig. \ref{fig:decision_pts}, each circle is a decision point and the bold green line denotes the updated reference route provided by the SP at each decision point $d_k^m(t_k^m)$. Starting from the origin $O_k$ at time $t^0_k$
% {\color{blue} Preferably $t_k^0$}
, when the non-compliant (red) vehicle $k\in V_{NC}(t^0_k)$ deviates from the reference route $\Pi^{ref}_k(t^0_k):=\{O_k,d^1_k(t^0_k),d^2_k(t^0_k),...,d^7_k(t^0_k),D_k\}$, a new reference route $\Pi^{ref}_k(t^1_k):=\{d^1_k(t^1_k),d^2_k(t^1_k),...,d^4_k(t^1_k),D_k\}$ will be updated by the SP before it arrives at its subsequent decision point $d^1_k(t^1_k)$. By continuously monitoring the actions of a vehicle at decision points, and updating the reference route when a non-compliant vehicle deviates from it, at time $t_k^m$ the updated reference route is given by $\Pi^{ref}_k(t_k^m):=\{d_k^m(t_k^m),d^{m+1}_k(t_k^m),...,D_k\}$, and the actual path up to $d_k^m(t_k^m)$ is defined as
\begin{align}
    \label{def:real_path}
    \Pi_k(t_k^m)=\{O_k,d^1_k(t^1_k),d^2_k(t^2_k),...,d_k^m(t_k^m)\}.
\end{align}

At each decision point $d_k^m(t_k^m)$, when a deviation of vehicle $k\in V_{NC}(t_k^m)$ is detected by the SP, the node $d_k^m(t_k^m)$ is considered as the new origin and a new reference route can be determined by solving the following optimization problem:

{\small\begin{align}
    \label{opt:social_guidance_tk}
    J^{ref}_k(t_k^m):=\min_{x^{ij}_k(t_k^m)} \sum_{i=1}^{|\mathcal{V}|} \sum_{j=1}^{|\mathcal{V}|} s^{ij}_k(t_k^m) x^{ij}_k(t_k^m),~s.t.~ \eqref{conseq:balanced_flow}-\eqref{conseq:origin_flow_inverse}
\end{align}}
where $s^{ij}_k(t_k^m)$ is the travel cost for vehicle $k$ along edge $(i,j)\in \mathcal{E}$ at time $t_k^m$, and $x^{ij}_k(t_k^m)$ is a binary variable, as in (\ref{opt:social_guidance}), defined at time $t_k^m$.
Essentially, the optimal routing problem \eqref{opt:social_guidance_tk} is the same as \eqref{opt:social_guidance} and the only difference is that all the variables in \eqref{opt:social_guidance_tk} are defined at time $t_k^m$. 
Therefore, the updated reference route at time $t_k^m$ is given by
\begin{align}
\label{def:ref_tk}
R^{ref}_k(t_k^m):=\bigcup_{(i,j)\in\mathcal{E}} \{e^{ij}\ |\ x^{ij}_k(t_k^m)=1\}.
\end{align}
Letting $N^{ref}_k(t_k^m)$ denote the cardinality of the set $R^{ref}_k(t_k^m)$ in \eqref{def:ref_tk}, the corresponding sequence of decision points along the reference $R^{ref}_k(t_k^m)$ is given by
\begin{align}
    \Pi_k^{ref}(t_k^m):=(d_k^m(t_k^m),d^{m+1}_k(t_k^m),...,d^{N_k^{ref}(t_k^m)-1}_k(t_k^m),D_k).
\end{align}

\subsection{Compliance Probability}
At the $m$-th decision point $d_k^m(t_k^m)$, a vehicle $k\in V_{NC}(t_k^m)$ has two choices: $i)$ follow the socially optimal reference route $R^{ref}_k(t_k^m)$ assigned by the SP, or $ii)$ adhere to its selfishly optimal routes, given by problem \eqref{opt:selfish_opt_init}, to minimize its own travel costs. Once the reference route $\Pi^{ref}_k(t^{m-1}_k)$ is provided by the SP at time $t^{m-1}_k$, the compliant behavior of vehicle $k$ can be measured by checking if the subsequent decision point belongs to the reference route, i.e., whether or not
\begin{align}
\label{check:decision_pt_tk}
    d_k^m(t^{m}_k)=d_k^m(t^{m-1}_k),
\end{align}
where $d_k^m(t^{m-1}_k)\in\Pi^{ref}_k(t^{m-1}_k)$. Define a binary variable $\zeta_k(t_k^{m-1})$ to describe the compliance behavior of vehicle $k$ at time $t_k^{m-1}$:
% {\color{blue} Please see the more formal def of $\zeta$}
\begin{align}
\label{zeta}
\zeta_k(t^{m-1}_k) = \textbf{1}
    [d_k^m(t^{m}_k)=d_k^m(t^{m-1}_k)],
\end{align}
where $\textbf{1}[\cdot]$ is the usual indicator function, i.e., $\zeta_k(t^{m-1}_k)=1$ when vehicle $k$ is considered compliant at time $t^{m-1}_k$, otherwise vehicle $k$ is non-compliant at time $t^{m-1}_k$.
% {\color{blue} It should be that $\zeta_k(t^{m-1}_k)=1$ when $k$ is compliant! This means using $1-\zeta$ in what follows. If you prefer, please just change (14) to 
% $\zeta_k(t^{m-1}_k) = \textbf{1}
%     [d_k^m(t^{m}_k) \ne d_k^m(t^{m-1}_k)]$
% and we can still use $\zeta_k(t^{m-1}_k)=0$ for compliance.
% }

For non-compliant HDVs, we design a cooperative compliance control scheme to incentivize them to align their behaviors with the reference routes through refundable tolls. This is done through a control $u_k(t_k^m)$ that represents the non-compliance cost issued by the SP and applied to the non-compliant vehicle $k\in V_{NC}(t_k^m)$ by deducting an amount $u_k(t_k^m)$ from its digital wallet. 
%if it deviates from the reference route $R^{ref}_k(t_k^m)$ at time $t_k^m$.
%{\color{blue} We need to specify that $u_k(t_k^m)>0$ and bounded by the ``refundable toll'' imposed on $k$ at time $t_k^0$. This is important because it means that the selfish term in (17) cannot take arbitrarily small values corresponding to $u_k^m$ becoming too large. Moreover, we need to keep track of the digital wallet value to make sure it does not go negative -- do you do that in the implementation?}
%{\color{red} Thanks, professor. I have specified the value of $u_k$ after eq. (15). Theoretically, for the bounded $u_k$, we still cannot make sure the digital wallet value to be non-negative if a human driver is always non-compliant. To avoid having negative values in the wallet, we can set some principles for those bad users to reset their digital wallet (e.g., buy the tokens back by using the real money, re-take the ``compliance" exams), or forbid them to benefit from social public facilities (e.g., forbid them from using highways). I added a Remark 1 to clarify this.}

%In general, human drivers prefer routes with lower travel costs for themselves. Considering that the non-compliant behaviors of the vehicle $k\in V_{NC}(t_k^m)$ will be penalized by $u_k(t_k^m)$ at each decision point $d_k^m(t_k^m)$, the total travel cost for vehicle $k$ at time $t_k^m$ to adhere its selfishly optimal route becomes:
To affect how a human driver $k$ decides to comply or not to the reference route $R_k^{ref}$, we modify the driver's cost function as follows:
\begin{align}
\label{eq:real_cost}
    J_k^{self}(t_k^m, u_k) = J_k(t_k^m) + \alpha_k (M^p_k(t_k^m)+u_k(t_k^m)),
\end{align}
where $J_k(t_k^m)$ denotes the path cost obtained from problem \eqref{opt:selfish_opt_init}.
Here, $u_k(t_k^m)>0$ and is bounded by the ``refundable toll'' imposed on $k$ at time $t_k^0$ to prevent $J_k^{self}(t_k^m,u_k)$ from becoming excessively large, while
$M^p_k(t_k^m)$ is the accumulated control (digital wallet deductions) over the time interval $[t^0_k,t_k^m]$ determined by
\vspace{-2mm}
\begin{align} 
\label{eq:sum_u}
M_k^p(t^{m}_k)=\sum_{i=1}^{m-1}(1-\zeta_k(t^i_k))u_k(t^i_k),
\end{align}
where $\zeta_k(t^i_k)$ was defined in (\ref{zeta}). Moreover, 
$\alpha_k$ denotes a weight to measure the importance of the cost controls relative to the selfish path cost evaluated by $k$ 
(a small value indicates a ``stubborn'' driver insensitive to tolls imposed by the SP). 
Note that we impose a simple additive control in (\ref{eq:real_cost}), however, this formulation can be modified to include nonlinear forms, allowing for more flexible and complex strategies.

\begin{rem}
Even though the control $u_k$ is positive and bounded by the initial value of the digital wallet, there is still no guarantee that the wallet balance will always remain non-negative because some drivers may be consistently non-compliant, leading to continuous token deductions and potentially exhausting their wallets over time. To prevent negative balances, the SP can set specific rules for resetting the digital wallets of such ``bad" users, e.g., penalizing them with actual currency, eliminating them from the system %{\color{blue} I left this out because I am not sure what such an ``exam'' is: or retake a compliance exam.} 
%{\color{red} I assume all users have to pass a compliance exam (similar to the driver's license exam) to have the right to drive vehicles and participate in the framework, but it's too idealistic and hard to achieve. I agree to remove this term here tp avoid any confusion.}
or applying restrictions to the access of public infrastructure (e.g., highway usage) to incentivize better compliance. 
    $\hfill\blacksquare$
\end{rem}

The compliance probability is defined as the likelihood that a vehicle complies with the reference route assigned by the SP. 
It is reasonable to assume that for an HDV $k$, this probability is influenced by the relative difference between the reference route $R_k^{ref}$ cost and its selfishly optimal route cost as modified in (\ref{eq:real_cost}) with digital wallet deductions due to non-compliant behavior. If the reference path cost $J^{ref}_k$ is significantly higher than the selfishly optimal cost $J^{self}_k$, the non-compliant HDV is less likely to follow the route guidance provided by the SP. Conversely, if a non-compliant HDV has incurred large amounts of token deductions, it becomes more likely to comply with the reference route $R^{ref}_k$ assigned by the SP to avoid further deductions.

Therefore, the value of the compliance probability $P_k(R^{ref}_k(t_k^m))$ satisfies two key properties: $i)$ $P_k(R^{ref}_k(t_k^m))$ decreases as $J^{ref}_k$ increases; and $ii)$ $P_k(R^{ref}_k(t_k^m))$ increases as $u_k(t_k^m)$ (or $M^p_k(t_k^m)$) increases. 
The exact form of $P_k(R^{ref}_k(t_k^m))$ is obviously unknown and this is complicated by the fact that it is also time-varying; for instance, a driver may start as highly non-compliant on a route and gradually become compliant because of the effect of $u_k(t_k^m)$ after several decision points; or after being penalized several times, the driver's sensitivity $\alpha_k$ to the tolls may also change. As we will see, this exact form is not crucial as long as the two structural properties above are satisfied so as to capture the driver's trade-off between favoring their selfish route and complying with the socially optimal route. In what follows, we will adopt a form similar to the Boltzmann distribution in \cite{rowlinson2005maxwell},
and assume that the probability of HDV $k\in V_{NC}(t_k^m)$ following the reference route $R^{ref}_k(t_k^m)$ is given by
\begin{align}
\label{eq:comp-Prob}
    P_k(R^{ref}_k(t_k^m),u_k):= \dfrac{e^{- J^{ref}_k(t_k^m)}}{e^{- J^{ref}_k(t_k^m)}+e^{-J^{self}_k(t_k^m,u_k)}}.
\end{align}
%After experiencing the effects of cooperative compliance control and paying tolls for non-compliant behaviors, vehicles may be more willing to follow guidance even when the reference route has higher travel costs than the selfish route. 
Here, $J^{ref}_k(t_k^m)$ is the travel cost for vehicle $k$ along reference route $R_k^{ref}(t_k^m)$ at time $t_k^m$, and $J^{self}_k(t_k^m)$ is the transformed through (\ref{eq:real_cost}) travel cost for vehicle $k$ along a selfish route, considering both past deducted tokens $M_k^p(t_k^m)$ and the potential future penalty $u_k(t_k^m)$.
%{\color{blue} How are ``future'' penalties captured in (15)? I don't see that.}
%{\color{red}In (15), $M_k^p$ is the past penalties, and $u_k$ is the future cost, it means that $k$ will be penalized $u_k$ if it is non-compliant at this point.}

It is easy to verify that the two properties mentioned above are satisfied by (\ref{eq:comp-Prob}). For instance, note that when $J_k^{ref}$ is very large, the compliance probability approaches 0; on the other hand, if $J^{self}_k(t_k^m,u_k)$ is large due to a large control value, the compliance probability may approach 1, depending on the control upper bound. 
% A similar expression of the compliance probability, as given in (\ref{eq:comp-Prob}) without any control, is also used in [{\color{blue}online learning CDC paper}], demonstrating consistency in modeling compliance probability across related works.

% {\color{blue} This is also a good place to point out that a similar expression as (17) is also used in our companion CDC paper (where exactly) but without the control.}

%However, in \eqref{eq:comp-Prob}, the path costs $J_k$ and $J^{ref}_k$ are time-varying, as they depend on the actual traffic flow within the road network and the actual position of vehicle $k$, where the traffic flow is influenced by the behaviors of other non-compliant vehicles. As a result, the compliance probability in \eqref{eq:comp-Prob} is also time-varying, evolving based on the interactions between vehicles and changing traffic conditions.
As already pointed out, 
in \eqref{eq:comp-Prob} the path costs $J_k^{self}$ and $J^{ref}_k$ are time-varying, as they depend on the actual traffic flow within the road network and the actual position of vehicle $k$, where the traffic flow is influenced by the behaviors of other non-compliant vehicles. 
This variation makes it challenging for the SP to effectively design a cooperation compliance control $u_k(t_k^m)$ for each vehicle $k$ at time $t_k^m$, aiming to increase the compliance probability while avoiding excessive penalties on non-compliant vehicles. 
In effect, the compliance probability can be viewed as a \emph{state variable} associated with vehicle $k$ that follows unknown dynamics. Estimating this state and driving it to a desirable value (ideally 1) is a task that we pursue next.
%Given that the dynamic model of compliance probability is unknown to the SP, we estimate its dynamics based on experience and 
To achieve this, we employ Control Lyapunov Functions (CLFs) to ensure that the actual compliance probability converges to a desired value $Q_{k}$ (ideally 1) set by the SP. 
% As a complement to this approach, appropriate Machine Learning (ML) methods can also be leveraged to learn the exact form of this dynamic model for each vehicle $k \in \mathcal{K}$ an approach followed in {\color{blue} [cite Online learning paper]}.

\begin{rem}
    As already pointed out, the determination of the compliance probability in \eqref{eq:comp-Prob} is not unique as long as the chosen form satisfies the two properties given above. In the routing problem, since $J^{ref}_k$ and $J^{self}_k$ represent the path costs for the reference route and selfishly optimal route (including penalties in the form of digital wallet deductions) respectively, these two values are bounded, making $P_k(R^{ref}_k(t_k^m)) \in (0,1)$. 
    %{\color{blue} Why not $[0,1]?$} {\color{red}Because the costs $J^{self}_k$ and $J^{ref}_k$ are bounded. When the difference between these two costs is large, the probability $P_k$ approaches to 0 or 1, but it never exactly equals 0 or 1.}
    This aligns with real-world traffic, as it is unrealistic to assume that a human driver will either always or never follow route guidance.
    $\hfill\blacksquare$
\end{rem}

% The precise nature of this model is not critical to our approach, but it obviously affects the overall performance of the control systems. There is an opportunity for appropriate Machine Learning methods to learn its exact form for each vehicle $k\in \mathcal{K}$.

\subsection{Cooperation Compliance Control}
In this section, we first briefly introduce Control Lyapunov Functions (CLFs), then illustrate how to employ CLFs to determine the control $u_k(t_k^m)$, $k\in V_{NC}(t_k^m)$, $m=0,1,2,...,$ to increase the compliance probability in \eqref{eq:comp-Prob} to some desired value $Q_{k}$.

\subsubsection{Control Lyapunov Functions}
\begin{defn}
(\emph{Class $\mathcal{K}$ Function} \cite{khalil2002nonlinear}) A continuous function $\alpha:[0,a)\rightarrow [0,\infty),a>0$ is said to belong to class $\mathcal{K}$ if it is strictly increasing and $\alpha(0)=0$.
\end{defn}

Consider an affine control system
\begin{equation}
\label{eq:affine_system}
    \dot{\bm{x}} = f(\bm{x}) + g(\bm{x})
    \bm{u},
\end{equation}
where $\bm{x}\in \mathbb{R}^n, \bm{u}\in \mathcal{U}\subset \mathbb{R}^q$ denote the state and control vector respectively, $f:\mathbb{R}^n \rightarrow \mathbb{R}^n$ and $g:\mathbb{R}^n \rightarrow \mathbb{R}^{n\times q}$ are Lipschitz continuous. 

\begin{defn}
\label{def:clf}
    (\emph{Control Lyapunov Function} \cite{ames2012control,xiao2023safe})
    % {\color{blue} I suggest using our own CBF book as a reference or at least add it on and use it later as well}
    A continuously differentiable function $V:\mathbb{R}^n \rightarrow \mathbb{R}$ is an exponentially stabilizing Control Lyapunov Function (CLF) for system (\ref{eq:affine_system}) if there exists constants $c_1>0,c_2>0,c_3>0$ such that for all $\bm{x}\in \mathbb{R}^n$, $c_1 ||\bm{x}||^2\leq V(\bm{x})\leq c_2||\bm{x}||^2$,
    \begin{equation} \label{eqn:clf}
        \inf_{\bm{u}\in \mathcal{U}} [L_fV(\bm{x})+L_gV(\bm{x})\bm{u}+c_3V(\bm{x})] \leq 0.
    \end{equation}
\end{defn}
% {\color{blue} No relaxation variable used?}

\subsubsection{Compliance Probability Convergence Analysis}
In real traffic conditions, it is difficult for the SP to accurately model how humans respond to controls (refundable tolls) in a complex and dynamic environment, making the dynamics of the compliance probability $P_k(R^{ref}_k(t))$ unknown to the SP. Let $\hat P_k(t)$ be the estimated compliance probability of vehicle $k\in V_{NC}(t_k^m)$ over time. Then, adopting the general form (\ref{eq:affine_system}) for the dynamics of the estimated compliance probability $\hat P_k(t)$ we have 
\begin{align}
\label{eq:phat_dynamics}
    \dot{\hat{P}}_k=f_k(\hat P_k)+g_k(\hat P_k)u_k,
\end{align}
where $f_k,g_k$ are Lipschitz continuous functions designed by the SP (based on experience or vehicle driving records), $u_k\in U$ is the control to vehicle $k$, and $U$ is the feasible control set, where the lower bound is 0 and the upper bound is determined by the initial value of the digital wallet.
% {\color{blue} We should mention that this is actually bounded by 0 and an upper bound given by the digital wallet initial value.}
The estimated compliance probability dynamics \eqref{eq:phat_dynamics} should satisfy the same two properties as \eqref{eq:comp-Prob} to ensure a reliable estimation. This is ensured by restricting $g_k(\cdot)$ to be positive to increase $\hat P_k$ as $u_k$ increases, while $f_k(\cdot)$ allows for considerable generality and can be updated at each decision point based on an observed error as defined next, hence exploiting a feedback mechanism. 

At each decision point, the error between the estimated compliance probability $\hat P_k(t_k^m)$ and our model \eqref{eq:comp-Prob} 
% {\color{blue} I am not sure we can claim this is ``actual'' but it is the one that we use in our model}
is
\begin{align}
\label{eq:error}
e_k(t_k^m):=  P_k(R^{ref}_k(t_k^m)) - \hat P_k(t_k^m).
\end{align}
Thus, the $f_k$ function in \eqref{eq:phat_dynamics} can be updated by
\begin{align}
    f_k(\hat P_k(t_k^m)) = f_k(\hat P_k(t^{m-1}_k)) + \dot e_k(t_k^m).
\end{align}
In this way, we always have measurements such that $e_k(t_k^m)=0$ and $\dot e(t_k^m)$ is close to 0 at $t_k^m$ by setting $P_k(R^{ref}_k(t_k^m)) = \hat P_k(t_k^m)$. This iteration increases the accuracy of the compliance probability estimation at each decision point. 
%{\color{blue} You did not use this feature in the next section. Have you tried it? Is it possible to add a version that includes the $f$ term in the simulation results and see what effect it has?}
%{\color{red} In the simulation, $f_k(0)$ is initialized as 0 at time $0$, and it will be updated at each $t_k^m$. }
 
The SP has two objectives: first, to ensure that the estimated compliance probability $\hat P_k(t)$ closely matches $P_k(R^{ref}_k(t))$, and second, to drive the compliance probability $P_k(R^{ref}_k(t))$ toward its desired value $Q_{k}$ as $t$ increases. To achieve these two objectives, we construct a Lyapunov function for vehicle $k\in \mathcal{K}$ as follows
\begin{align}
\label{eq:Lyapunov_func}
    V_k(\hat P_k)=\xi_1(\hat P_k-Q_{k})^2+\xi_2(\hat P_k-P_k(R^{ref}_k))^2,
\end{align}
where $\xi_1,\xi_2\in(0,1)$ denote weights forming a convex combination above and $Q_{k}\in[0,1]$ is a desired compliance probability for vehicle $k\in \mathcal{K}$. We aim for $\hat P_k(t_k^m)$ to converge to $P_k(R^{ref}_k(t_k^m))$ and for $\hat P_k(t_k^m)$ to converge to $Q_k$.
%this joint process induces the convergence of $\hat P_k(t_k^m)$ to $Q_k$ as $m$ increases.
% $P_k(R^{ref}_k)$ is the real compliance probability of vehicle $k$ in (\ref{eq:comp-Prob}) to comply with the reference route. 
%{\color{blue} There is something unclear here: we have dropped the time dependence from $P_k(R^{ref}_k(t_k^m)$. This needs to be better explained. What value does $\hat P_k$ converge to exactly?}
Once the Lyapunov function \eqref{eq:Lyapunov_func} is proved to be decreasing and converging to 0, the two objectives of the SP can be achieved, ensuring that the actual compliance probability $P_k(R^{ref}_k)$ of vehicle $k$ eventually reaches the desired value $Q_k$.

% Motivated by Definition \ref{def:clf}, we define a set
% {\small
% \begin{align}
% \label{set:clf_control}
%    K(\hat P_k)=\{u_k\in U:L_fV_k(\hat P_k)+L_gV_k(\hat P_k)u_k+c_3V_k(\hat P_k)\leq 0\}
% \end{align}
% }

% \noindent consisting of the controls that result in $\dot{V}_k(\hat P_k)\leq -c_3V_k(\hat P_k)$, which further indicates that
% \begin{align}
%     V_k(\hat P_k(t))\leq V_k(\hat P_k(t^0_k))e^{-c_3t}
% \end{align}
% where $V_k(\hat P_k(t^0_k))$ denotes the initial value of the Lyapunov function at time $t^0_k$. 
% Hence, as proved in \cite{ames2014rapidly}, for any locally Lipschitz continuous feedback control law $u_k$ such that $u_k\in K(\hat P_k)$ for all $\hat P_k$, the Lyapunov function $V_k(\hat P_k)$ is an exponentially stabilizing Control Lyapunov Function, and converges to 0 as $t\rightarrow \infty$. 
\begin{thm}
(based on \cite{ames2014rapidly}) For the system \eqref{eq:phat_dynamics}, if any locally Lipschitz continuous feedback control law $u_k$ belongs to the set
\vspace{-1mm}
{\small
\begin{align}
\label{set:clf_control}
   K(\hat P_k)=\{u_k\in U:L_fV_k(\hat P_k)+L_gV_k(\hat P_k)u_k+c_3V_k(\hat P_k)\leq 0\}
\end{align}}
for all $\hat P_k$, the Lyapunov function $V_k(\hat P_k)$ is an exponentially stabilizing Control Lyapunov Function and converges to 0 as $t\rightarrow \infty$. 
\end{thm}
\begin{proof}
    The set $K(\hat P_k)$ consists of the controls that result in $\dot{V}_k(\hat P_k)\leq -c_3V_k(\hat P_k)$, which further indicates that
\begin{align}
    V_k(\hat P_k(t))\leq V_k(\hat P_k(t^0_k))e^{-c_3t},
\end{align}
where $V_k(\hat P_k(t^0_k))$ denotes the initial value of the Lyapunov function at time $t^0_k$. Hence, as proved in \cite{ames2014rapidly}, for any locally Lipschitz continuous feedback control law $u_k$ such that $u_k\in K(\hat P_k)$ for all $\hat P_k$, the Lyapunov function $V_k(\hat P_k)$ is an exponentially stabilizing Control Lyapunov Function, and converges to 0 as $t\rightarrow \infty$. 
\end{proof}

% {\color{blue} Can we state this more formally? That is, state a theorem asserting convergence where the proof simply points to [30] by using the LF $V_k(\hat P_k(t))$}
This result demonstrates that: first, the estimated compliance probability $\hat P_k$ converges to its desired value $Q_k$, and second, $\hat P_k$ converges to the actual compliance probability $P_k(R^{ref}_k)$. Consequently, this establishes that the actual compliance probability $P_k(R^{ref}_k)$ converges to the desired value $Q_k$.

The final step is the determination of the control $u_k$. Clearly, this has 
%, to incentivize the human drivers to comply with SP guidance and increase their actual compliance probability $P_k(R^{ref}_k)$ toward the desired $Q_k$, the control $u_k$ has 
to be selected from the set defined in \eqref{set:clf_control}.
We proceed in the standard way used to derive controls based on CLFs (see \cite{xiao2023safe})
by updating the control at each decision point $d_k^m(t_k^m)$ as the solution of the following quadratic program:
\begin{subequations}
\label{qp}
    \begin{align}
    \label{qp:obj}
    \min_{u_k(t_k^m),\delta_k(t_k^m)} &u_k^2(t_k^m)+\gamma\delta_k^2(t_k^m)\\
    s.t.~
    \label{eq:clf_condition}
    L_fV_k(\hat P_k)&+L_gV_k(\hat P_k)u_k+c_3V_k(\hat P_k)\leq \delta_k^2,
\end{align}
\end{subequations}
where $\delta_k$ is a controllable variable that relaxes \eqref{eq:clf_condition} as a soft constraint, $\gamma$ is an adjustable weight to capture the relative importance of the two terms in \eqref{qp:obj}, and $u_k(t_k^m)\in U$.
% {\color{blue} Note that (26b) relaxes (19) and it is the one that turns (19) into a soft constraint -- please check and reword accordingly}
Note that \eqref{eq:clf_condition} relaxes \eqref{eqn:clf} and it is the one that turns \eqref{eqn:clf} into a soft constraint to 
ensure the feasibility of the quadratic program \eqref{qp}.
% {\color{blue} This addresses my earlier comment on relaxation. I actually see this as an important part of the whole process because it addresses the ``fairness'' concern. So perhaps it is not just a Remark but another component of CCC.}

Moreover, for a road network, another objective of the SP is to maximize the \emph{average compliance probability} $Q^*=\frac{1}{N}\sum_{k=1}^N Q_k$ over all vehicles, bringing it as close to 1 as possible. Ideally, if all the HDVs respond positively to the control input, 
%$Q_k$ is set equal to $Q^*$, i.e.,
$Q_k=Q^*,\; \forall k\in \mathcal{K}$. However, when there exist some stubborn HDVs $k_s\in V_{NC_s}\subset \mathcal{K}$ that consistently refuse to comply with the SP routes, we deal with them by setting higher values for $Q_k$ if $k\in \mathcal{K}/V_{NC_s}$. In this manner, we compensate for these non-compliant HDVs, thereby maintaining the desired average compliance probability $Q^*$. In other words, for $a = \frac{|V_{NC_s}|}{|\mathcal{K}|}$ that represents the fraction of the stubborn non-compliant vehicles in the road network, the choice of $Q_k$ may be $Q_k=\min\{\frac{Q^*}{1-a},1\}, k\in \mathcal{K}/V_{NC_s}$. 
This addresses the \emph{fairness} issue that arises when some users of a system aiming at social optimality are not as compliant as others.

\subsection{Vehicle Coordination in the Routing Problem}
In this section, we describe the complete procedure of applying the cooperation compliance control for non-compliant vehicles in a road network $\mathcal{G}(\mathcal{V},\mathcal{E})$, 
%{\color{blue} It seems more like a summary of the whole procedure..}
%{\color{red}Yes, I hope to provide a clear map to the readers of how we implement the CCC framework in a routing problem.}
to increase the compliance probability of the non-compliant vehicles, and ultimately improve traffic performance. Assuming vehicle $k\in V_{NC}(t^0_k)$ starts at time $t^0_k$, the SP proceeds as follows:

1) Step 1: 
% {\color{blue} Let's make clear who exactly does what at every step. This first step is done by the SP.}
The SP collects the information for all vehicles $k$ in set $\mathcal{K}$, including: origin $O_k$, destination $D_k$, departure time $t^0_k$, and the travel cost $c^{ij}_k$ on all edges $(i,j)\in \mathcal{E}$. Then, it estimates the real-time flow and determines the social travel cost $s^{ij}_k$ on all edges $(i,j)\in \mathcal{E}$ accordingly. 

2) Step 2: At time $t^0_k$, the SP solves the optimization problem \eqref{opt:social_guidance_tk} and assigns the reference route $R^{ref}_k(t^0_k)$ to vehicle $k$. All CAVs are assumed to be fully compliant, with a compliance probability of 1, whereas HDVs may follow the reference route with a compliance probability of $P_k(R^{ref}_k(t^0_k))$, where $k \in V_{NC}(t^0_k)$, since their selfishly optimal route determined by \eqref{opt:selfish_opt_init} differs from the socially optimal one.

3) Step 3: After vehicles depart from their origins, the SP continuously monitors the position of all vehicles $k$ in set $\mathcal{K}$ and identifies the node number of its next decision point. Upon reaching each decision point $d_k^m(t_k^m)$, the SP checks whether the condition \eqref{check:decision_pt_tk} for vehicle $k\in V_{NC}(t_k^m)$ is satisfied.

4) Step 4: If the condition \eqref{check:decision_pt_tk} for vehicle $k\in V_{NC}(t_k^m)$ is satisfied, then there is no penalty for vehicle $k$ at time $t_k^m$, i.e., $u_k(t_k^m)=0, \zeta_k(t_k^m)=1$, and the compliance probability remains unchanged; otherwise, the SP solves the QP \eqref{qp} 
% {\color{blue} Who does that?}
to obtain the control $u_k(t_k^m)$ and update the compliance probability $P_k(R^{ref}_k(t_k^m))$ based on the current social travel cost $s^{ij}_k(t_k^m)$ and the control $u_k(t_k^m)$.

5) Step 5: If vehicle $k$ arrived at the destination, remove it from the road network, and update the real-time travel flow and the travel cost $s^{ij}_k$ until all vehicles have arrived at their destinations; otherwise, return to Step 3).

\section{SIMULATION RESULTS}
\label{sec:simulation}
\subsection{Network Setup}
This section presents simulation results demonstrating the effectiveness of the proposed CCC framework in the mixed-traffic routing problem. We consider a subset of the Boston urban road network, as shown in Fig. \ref{fig:routing}. The road network consists of 37 nodes and 58 edges, and the length of each edge is labeled in Fig. \ref{fig:routing} and is obtained from \textit{GoogleMaps} \cite{googlemap}. The simulation of this network includes 1000 vehicle trips, and their origins and destinations are randomly selected from the red nodes in the road network. The speed limit is set to $35$ mph for local road edges and $60$ mph for highway edges, and the capacity is 2000 vehicles per hour for each edge. All vehicles are assumed to depart from their origins simultaneously and make decisions at each node (decision point). The real travel time on each edge is modeled by the Bureau of Public Road (BPR) function \cite{united1964traffic}, with parameters $\alpha = 0.14,~\beta = 4$. 
\begin{figure}[hpbt]
    \centering   
    % \vspace*{-\baselineskip} 
    \includegraphics[ width=0.8\linewidth]{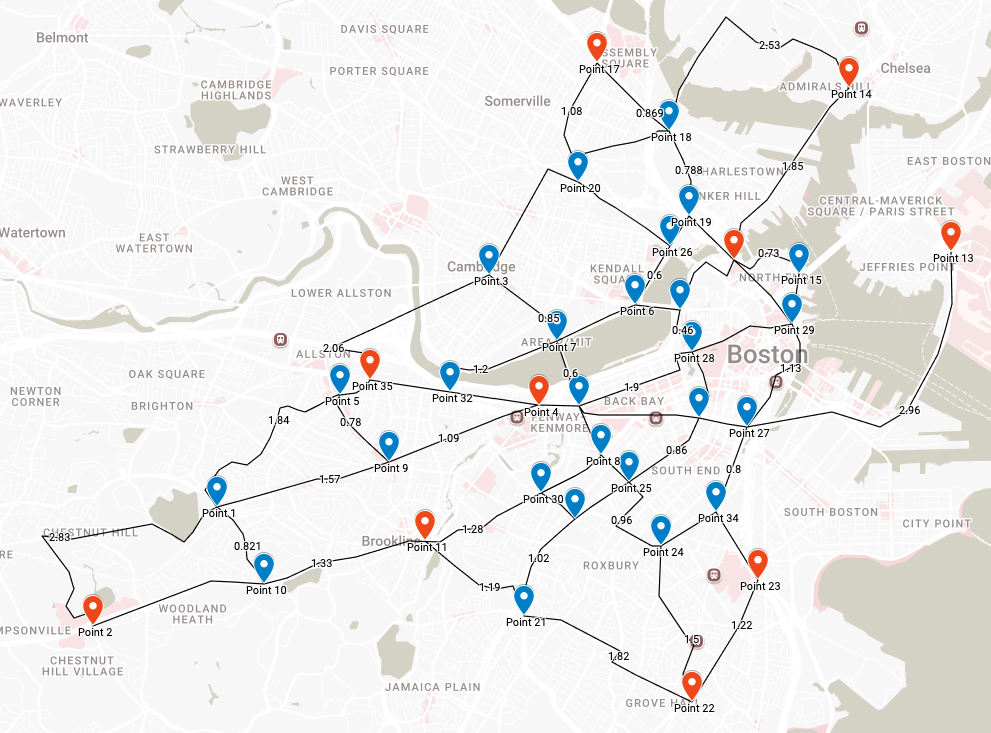} 
    % \vspace{-2mm}
    \caption{Boston area road network.}
    \label{fig:routing}
   % \vspace*{-\baselineskip}
   % \vspace*{-3mm}
\end{figure}

In the simulation, all vehicles are assumed to minimize their travel time for the trip. Given the origins and destinations, the SP estimates the real-time flow on each edge using the BPR model and assigns a socially optimal route to each vehicle to minimize overall travel time and mitigate congestion based on traffic conditions. In the optimization problem \eqref{opt:social_guidance_tk}, the travel cost $s^{ij}_k$ is set as the \textit{real flow} time on edge $e^{ij}$. 
However, since HDVs lack perfect information about other vehicles in the network, they may estimate their arrival time by using free-flow conditions. The travel cost $c^{ij}_k$ in problem \eqref{opt:selfish_opt_init} is set as the \textit{free flow} time on edge $e^{ij}$. This discrepancy can lead to traffic congestion and reduce the overall efficiency of the network.

\subsection{Results and Analysis}
The SP is a centralized coordinator that continuously monitors the position of all vehicles in the road network. Due to discrepancies between the routes determined by problems \eqref{opt:social_guidance_tk} and \eqref{opt:selfish_opt_init}, some HDVs may deviate from the reference route assigned by the SP and instead follow their own selfishly selected routes. At each decision point $d_k^m(t_k^m)$, the SP updates the travel cost $s^{ij}_k(t_k^m)$ for all edges $(i,j) \in \mathcal{E}$ based on real-time traffic flow information. New reference routes are then reassigned to all vehicles after detecting non-compliant behaviors
% {\color{blue} Why only non-compliant? If traffic conditions have changed, optimal routes may have changed for everyone.}
to minimize traffic disruptions. The time window for estimating the real traffic flow is $\Delta t=120s$. In \eqref{eq:real_cost}, the parameter is set as $\alpha_k=3$, and the dynamics of the estimated compliance probability in (\ref{eq:phat_dynamics}) are given by
% {\color{blue} I added $f_k$ for generality since you do eventually get nonzero values}
\begin{align}
\nonumber
    \dot{\hat P}_k = f_k(\hat P_k(t_k^m)) + u_k,
\end{align}
where $f_k(\hat P_k(0))=0$ and will be updated at each $t_k^m$, and we have set $g_k(t)=1$ for all $t$.
The weights in \eqref{eq:Lyapunov_func} are set as $\xi_1=0.5,\xi_2=0.5$ to balance the estimation error and desired value reachability. In \eqref{qp}, the parameters are $\gamma=0.5$, $c_3=100$. The desired compliance probability is set as $Q^*=0.9$ with $Q_k=0.9,~\forall k\in\mathcal{K}$.  
\begin{figure}[hpbt]
    \centering   
    % \vspace*{-17mm}
    % \begin{adjustbox}{width=\textwidth,center}
    \begin{subfigure}{\linewidth}
    \centering 
        %\begin{adjustbox}{width=0.33\textwidth,center}
      % include first image
      \includegraphics[width=\linewidth]{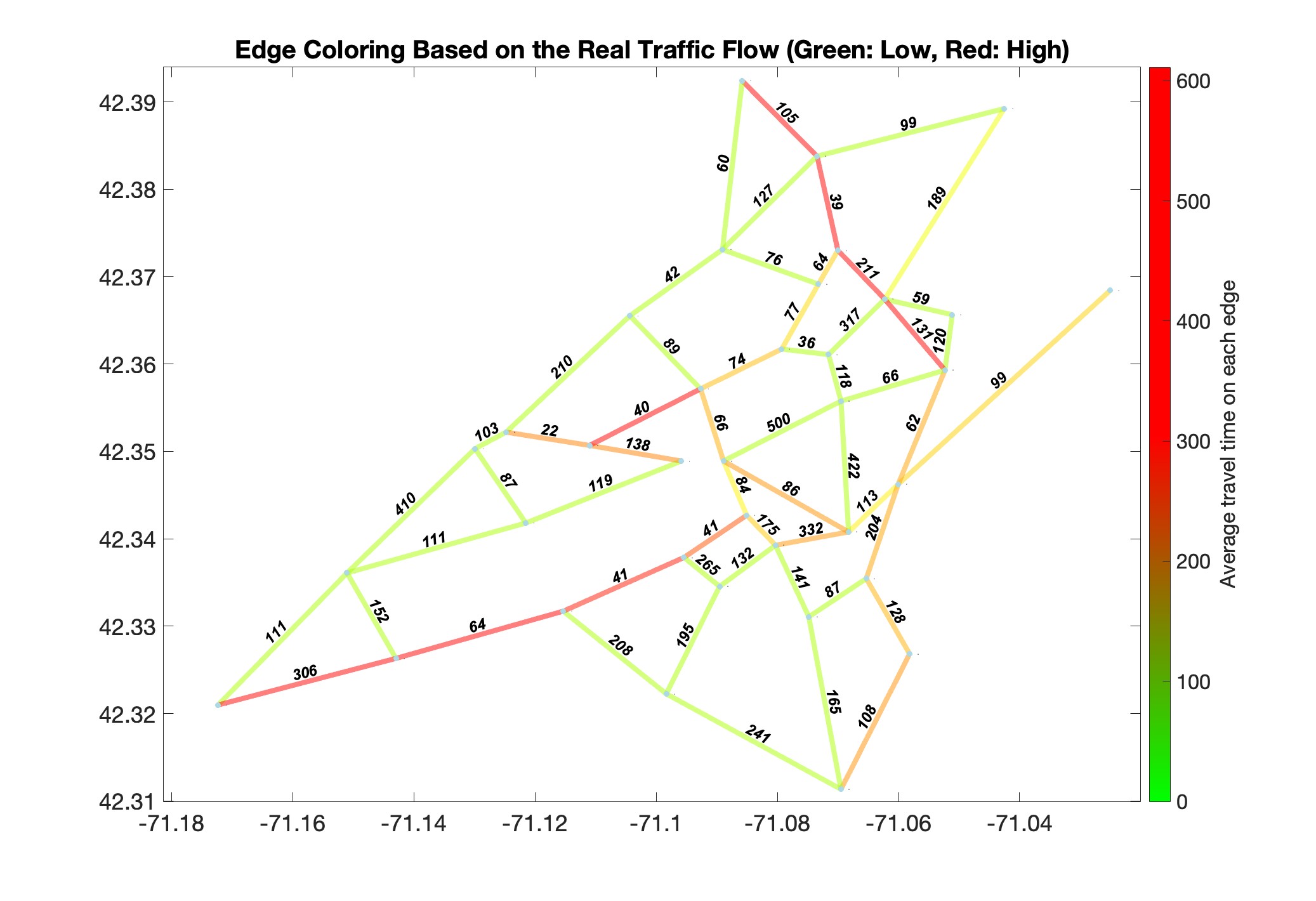}  
     \caption{Average travel time of each edge under Baseline.}
    \label{fig:baseline}
      %\end{adjustbox}
    \end{subfigure}   
    
    \begin{subfigure}{\linewidth}
    \centering 
      %  \begin{adjustbox}{width=0.33\textwidth,center}
      \includegraphics[width=\linewidth]{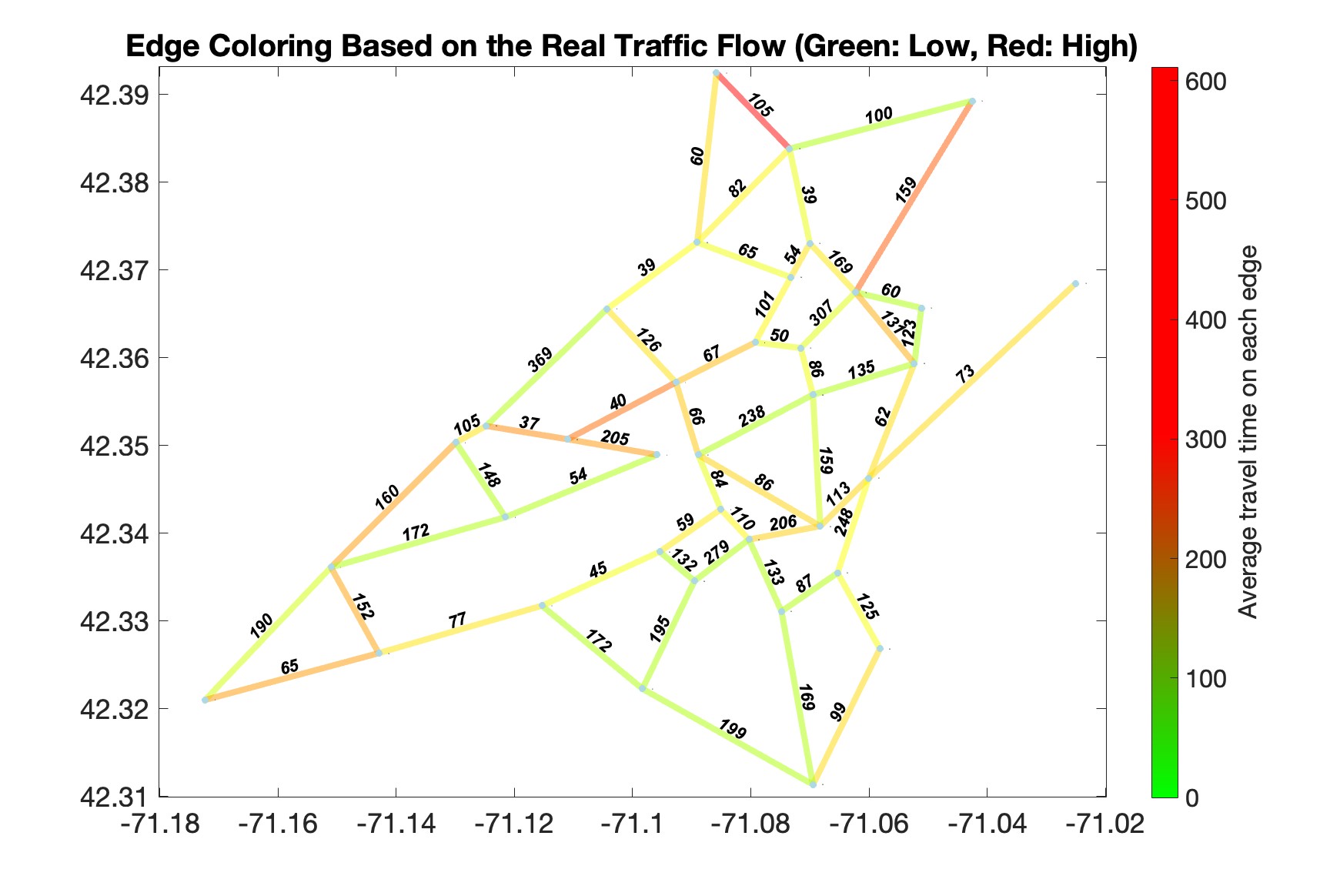}  
      \caption{\centering{Average travel time of each edge under CCC.}}
      \label{fig:ccc}
      %\end{adjustbox}
    \end{subfigure}
    % \end{adjustbox}
     \caption{Average travel time of each edge.}
     % \vspace*{-\baselineskip} %\vspace*{-\baselineskip}
\end{figure}

% \begin{figure}[hpbt]
%     \centering   
%     % \vspace*{-\baselineskip} 
%     \includegraphics[ width=\linewidth]{Figures/baseline.eps} 
%     % \vspace{-2mm}
%     \caption{Average travel time of each edge under baseline.}
%     \label{fig:baseline}
%    % \vspace*{-\baselineskip}
%    % \vspace*{-3mm}
% \end{figure}

% \begin{figure}[hpbt]
%     \centering   
%     % \vspace*{-\baselineskip} 
%     \includegraphics[ width=\linewidth]{Figures/map_comp.eps} 
%     % \vspace{-2mm}
%     \caption{Average travel time of each edge under CCC framework.}
%     \label{fig:ccc}
%    % \vspace*{-\baselineskip}
%    % \vspace*{-3mm}
% \end{figure}

We consider the selfish optimal planning problem given by \eqref{opt:selfish_opt_init} as the ``Baseline". The average travel time per vehicle decreased from 1342.9s in the baseline case to 1029.1s in the socially optimal case after the SP applied control $u_k(t_k^m)$ at each decision point where a vehicle deviates from the SP reference route, representing a 23.3\% reduction. After implementing the CCC framework, the maximum and minimum travel times are reduced to $1340.1s$ and $581.9s$, respectively, compared to $2129.7s$ and $805.6s$ in the baseline case. 

Additionally, the traffic congestion, which is captured by the average travel time along each edge, is illustrated in Fig. (\ref{fig:baseline}) for the Baseline case and in Fig. (\ref{fig:ccc}) under the CCC framework.
% {\color{blue} You may want to join these two figures into one to force them to show up together in the final version: it is important for the reader to be able to see them together and compare!}
In these two figures, green edges indicate that the actual average travel time is less than or equal to the free-flow time, red edges are those where the actual travel time exceeds twice the free-flow time, and orange edges are those where the actual average travel time is between the free-flow time and twice its value. In Fig. (\ref{fig:baseline}), there are 8 red edges and some green edges, indicating that some parts of the network are heavily congested while others remain underutilized. However, in Fig. (\ref{fig:ccc}), when the SP assigns socially optimal reference routes to each vehicle and implements cooperation compliance control to address non-compliant behaviors, the network utilization increases. As a result, there are no red edges and no heavy traffic congestion, indicating that the traffic distribution across the map is more balanced.

The evolution of the compliance probability as a function of the number of decision points is illustrated in Fig. \ref{fig:comp_prob}.
%where the X-axis represents the number of decision points, and the Y-axis denotes the compliance probability value. 
At each decision point, the control $u_k$ for vehicle $k$ is updated, leading to a corresponding update in $P_k(R_k^{ref})$. Note that vehicles reach decision points at different times in the routing problem, causing compliance probability updates to occur \emph{asynchronously} across different vehicles. 
As shown in Fig. \ref{fig:comp_prob}, the compliance probability has increased from 0 to approximately 0.9 for all vehicles after passing 5 decision points, demonstrating the effectiveness of the cooperative compliance control in \eqref{qp}. This result confirms that the actual compliance probability can converge to the desired value, even when its dynamics are unknown to the SP. 

Note that the convergence of the compliance probability is not required to be achieved within a single trip. The desired value $Q_k$ can be reached either within a single trip or over multiple trips, depending on the path length and the number of decision points encountered along the path. After each trip, the driver starts the next trip with the latest digital wallet balance and re-initializes the compliance probability based on \eqref{eq:comp-Prob}, causing the convergence process to evolve over a long-term horizon.

\begin{figure}[hpbt]
    \centering   
    % \vspace*{-\baselineskip} 
    \includegraphics[ width=0.9\linewidth]{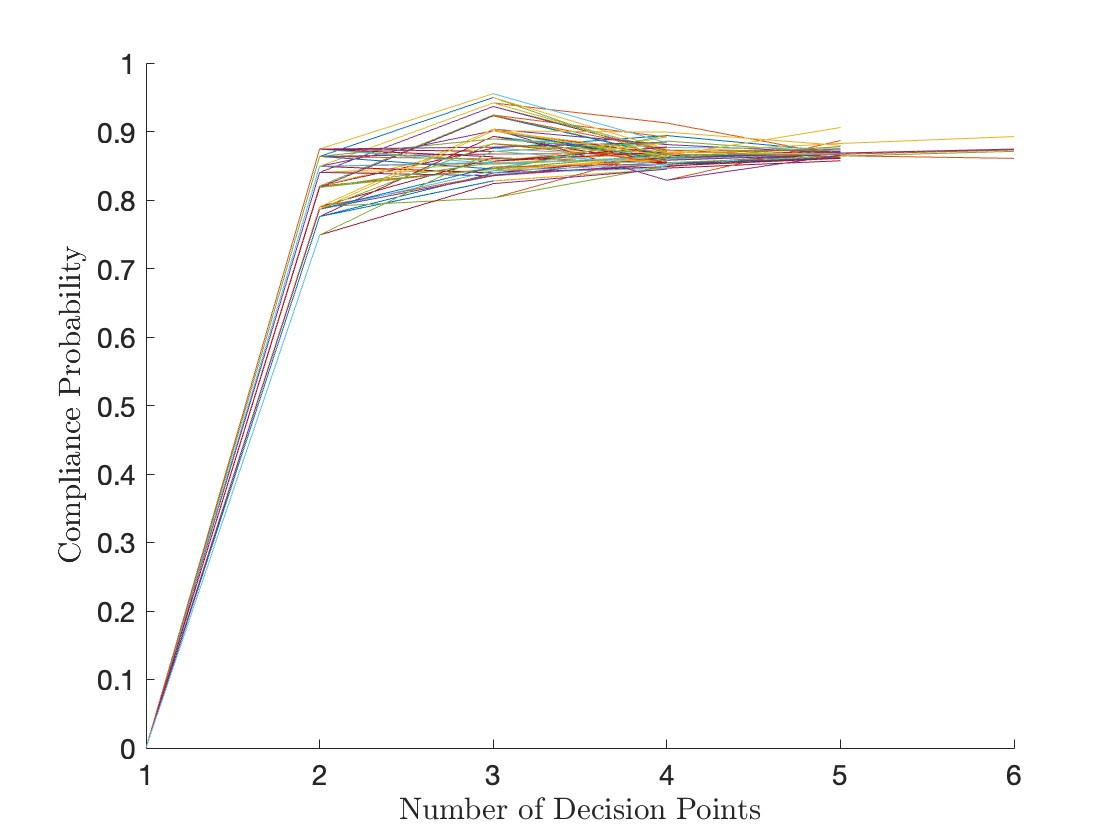} 
    % \vspace{-2mm}
    \caption{Compliance Probability Convergence.}
    \label{fig:comp_prob}
   % \vspace*{-\baselineskip}
   % \vspace*{-3mm}
\end{figure}

\section{CONCLUSIONS AND FUTURE WORK}
\label{secV:Conclusions}
In this paper, we presented a CCC framework designed for mixed traffic routing problems, aimed at incentivizing non-cooperative HDVs to adhere to the system-wide optimal route guidance established by a SP through a refundable toll scheme. To address the heterogeneous and unpredictable dynamics of human driver responses, we employed CLFs, enabling us to adaptively update our compliance probability model online and ensure the convergence of actual compliance probabilities to desired values, thereby enhancing the efficiency of the road network. The results of our simulations demonstrated the CCC framework’s effectiveness in reducing travel times for vehicle trips and alleviating traffic congestion. Future research will explore the application of advanced machine-learning techniques to accurately model human compliance probabilities based on historical data.

\bibliographystyle{IEEEtran}
\bibliography{cmp}

\end{document}